\documentclass[ARXIV,preprint]{imsart}

\RequirePackage[OT1]{fontenc}
\RequirePackage{amsthm,amsmath}
\RequirePackage[colorlinks,citecolor=blue,urlcolor=blue]{hyperref}

\usepackage{bbm}
\usepackage{algorithm2e}
\usepackage{enumitem}
\usepackage{multirow}
\usepackage{graphicx}
\usepackage{subcaption}

\usepackage{ amssymb }
\newcommand{\abs}[1]{\ensuremath{\bigl| #1 \bigr|}}

\newcommand{\normtwo}[1]{\ensuremath{|\!|#1|\!|_{2}}}
\newcommand{\normtwos}[1]{\ensuremath{|\!|#1|\!|_{2}^{2}}}

\usepackage{cite}

\makeatletter
\renewcommand{\@biblabel}[1]{(#1)}
\makeatother

\newcommand{\outcome}{\ensuremath{\boldsymbol{y}}}
\newcommand{\corr}[2]{\ensuremath{c\bigl[#1,#2\bigr]}}
\newcommand{\tuningset}{\ensuremath{\mathcal{R}}}

\newcommand{\identity}{\ensuremath{\boldsymbol{I}}}
\newcommand{\real}[1]{\ensuremath{\mathbb{R}^{#1}}}
\newcommand{\optimal}{\ensuremath{r^{*}}}
\newcommand{\tuning}{\ensuremath{r}}
\newcommand{\ridgetuning}{\ensuremath{t}}
\newcommand{\ridgetuningset}{\ensuremath{\mathcal{T}}}
\newcommand{\ridgetuningsetmapped}{\ensuremath{\tuningset_\phi}}

\newcommand{\noise}{\ensuremath{\boldsymbol{u}}}
\newcommand{\leftsingular}{\ensuremath{\boldsymbol{U}}}

\newcommand{\rightsingular}{\ensuremath{\boldsymbol{V}}}

\newcommand{\diagnal}{\ensuremath{\boldsymbol{D}}}
\newcommand{\diagnalinv}{\ensuremath{\boldsymbol{D^{\dagger}}}}
\newcommand{\singularvalue}{\ensuremath{d}}
\newcommand{\design}{\ensuremath{\boldsymbol{X}}}
\newcommand{\testing}{\ensuremath{\boldsymbol{z}}}

\newcommand{\target}{\ensuremath{\boldsymbol{\beta}}}
\newcommand{\true}{\ensuremath{\boldsymbol{\beta}^{*}}}
\newcommand{\estimator}{\ensuremath{\boldsymbol{\hat{\beta}}}}
\newcommand{\regularizedestimator}[2]{\ensuremath{\boldsymbol{\hat{\beta}}_{#1}\bigl[#2\bigr]}}

\newcommand{\avtuning}{\ensuremath{\hat{\tuning}}}
\newcommand{\avtuningridge}{\ensuremath{\hat{\ridgetuning}}}
\newcommand{\avtuningset}{\ensuremath{\tuningset_{A}}}
\newcommand{\tuningone}{\ensuremath{\tuning^{\prime}}}
\newcommand{\tuningtwo}{\ensuremath{\tuning^{\prime\prime}}}
\newcommand{\oracletuning}{\ensuremath{{\tuning_{\vspace{-0.3mm}o}}}}
\newcommand{\bitesting}{\ensuremath{\hat{s}}}
\newcommand{\tuningmapping}[1]{\ensuremath{\phi [ #1 ]}}

\newcommand{\errorprobability}{\ensuremath{\delta}}
\newcommand{\gaussiantuning}{\ensuremath{\tuning_\errorprobability}}

\newcommand{\nameedr}{\ensuremath{\operatorname{edr}}}
\newcommand{\nameridge}{\ensuremath{\operatorname{ridge}}}
\newcommand{\pavedr}{\ensuremath{\operatorname{PAV_{\nameedr}}}}

\DeclareMathOperator*{\argmin}{arg\,min}

\newcommand{\probability}[1]{\ensuremath{\mathbb{P}\bigl\{ #1 \bigr\} }}

\newcommand{\variance}[1]{\ensuremath{\operatorname{Var}\bigl[ #1 \bigr] }}
\newcommand{\normalDist}[2]{\ensuremath{\mathcal{N}[#1,#2]}}

\numberwithin{equation}{section}
\theoremstyle{plain}
\newtheorem{thm}{Theorem}[section]
\newtheorem{lemma}{Lemma}[section]

\newtheorem{example}{Example}[section]

\newtheorem{definition}{Definition}[section]

\newcommand{\ra}[1]{\renewcommand{\arraystretch}{#1}}

\startlocaldefs
\numberwithin{equation}{section}
\theoremstyle{plain}
\endlocaldefs

\begin{document}

\begin{frontmatter}
\title{Tuning parameter calibration for prediction in personalized medicine}
\runtitle{Tuning parameter calibration for prediction in personalized medicine}

\begin{aug}
\author{\fnms{Shih-Ting} \snm{Huang}\thanksref{m1}\ead[label=e1]{shih-ting.huang@rub.de}},
\author{\fnms{Yannick} \snm{D\"uren}\thanksref{m1}\ead[label=e2]{yannick.dueren@rub.de}},
\author{\fnms{Kristoffer H.} \snm{Hellton}\thanksref{m2}\ead[label=e3]{hellton@nr.no},
\ead[label=u1,url]{www.mn.uio.no/math/english/people/aca/kristohh/}}
\and
\author{\fnms{Johannes} \snm{Lederer}\thanksref{m1}
\ead[label=e4]{johannes.lederer@rub.de},\ead[label=u2,url]{www.johanneslederer.com}
}

\runauthor{S.-T. Huang et al.}

\affiliation{Ruhr-University Bochum\thanksmark{m1} and Norwegian Computing Center\thanksmark{m2}}

\address{Ruhr-University Bochum\\
Universit\"atsstra{\ss}e 150\\
44801 Bochum \\
Germany \\
\printead{e1}\\
\phantom{E-mail:\ }\printead*{e2}\\
\phantom{E-mail:\ }\printead*{e4}\\
\printead{u2}}

\address{Norwegian Computing Center\\
P.O.\@ Box 114 Blindern \\
0314 Oslo \\
Norway \\
\printead{e3}\\
\printead{u1}}
\end{aug}

\begin{abstract}
Personalized medicine is becoming an important part of medicine, for instance predicting individual drug responses or risk of complications based on genomic information. However, many current statistical methods are not tailored to this task, because they overlook the individual heterogeneity of patients. In this paper, we look at personalized medicine from a linear regression standpoint. We introduce an alternative version of the ridge estimator and target individuals by establishing a tuning parameter calibration scheme that minimizes prediction errors of individual patients. In stark contrast, classical schemes such as  cross-validation minimize prediction errors only on average. We show that our pipeline is optimal in terms of oracle inequalities, fast, and highly effective both in simulations and on real data.
\end{abstract}


\begin{keyword}
\kwd{personalized medicine}
\kwd{precision medicine}
\kwd{personalized prediction}
\kwd{tuning parameter calibration}
\kwd{euclidean distance ridge}
\kwd{ridge regression}
\kwd{linear regression}
\kwd{regularization}
\end{keyword}

\end{frontmatter}

\section{Introduction}

In the last decade, improvements in genomic, transcriptomic, and proteomic technologies have enabled personalized  medicine (also called precision medicine) to become an essential part 
of contemporary medicine.
%
Personalized medicine takes into account individual variability in genes, proteins, environment, and lifestyle to decide on optimal disease treatment and prevention~\cite{hamburg2010path}. The use of a patient's genetic and epigenetic information has already proven to be highly effective to tailor drug therapies or preventive care in a number of applications, such as breast cancer ~\cite{cho2012personalized}, prostate cancer~\cite{nam2007assessing}, ovarian cancer~\cite{hippisley2015development}, and pancreatic cancer~\cite{ogino2011cancer}, cardiovascular disease~\cite{ehret2011genetic}, cystic fibrosis~\cite{waters2018human}, and psychiatry~\cite{demkow2017genetic}. 
The subfield of pharmacogenomics studies specifically how genes affect a person's response to particular drugs to develop more efficient and safer  medications~\cite{ziegler2012personalized}. 


Genomic, epigenomic, and transcriptomic data used in precision medicine, such as gene expression, copy number variants, or methylation levels are typically high-dimensional with a number of variables that rivals or exceeds the number of observations. Using such data to estimate and predict treatment response or risk of complications, therefore, requires regularization typically by the $\ell_1$ norm (lasso), the $\ell_2$ norm (ridge), 
or other terms. Ridge regression~\cite{hoerl1970ridge} yields good predictive performance for dense or non-sparse effects, that is, for outcomes related to systemic conditions, as the method does not perform variable selection. Ridge regression has become a standard tool for prediction based on genomic data, and it has been shown that ridge regression can outmatch competing prediction methods for survival based on gene expression~\cite{bovelstad2007predicting,cule2013ridge}.

However, regularization always introduces one or more tuning parameters. These tuning parameters are usually calibrated based on the averaged prediction risks. 
Most commonly used, $K$-fold cross-validation (CV) divides the data into $K$~folds (typically $K \in \{5, 10 \}$), predicts each fold out-of-sample, averages over all folds for a range of tuning parameters, and selects the value with the lowest averaged error~\cite{stone1974cross,golub1979generalized}. 
But the averaging removes the inherent individual heterogeneity of the patients and can, therefore, result in sub-optimal prediction performance. This may ultimately lead to unsuitable treatment, administration of improper medication with adverse side effects, or lack of preventive care~\cite{hamburg2010path}. 

Hence, rather than minimizing an averaged prediction error, our goal is to minimize each patient's individual (``personalized'') prediction error. A na\"ive two-stage personalized procedure for ridge regression was recently proposed by~\cite{hellton2018fridge}. 
In this paper, we introduce an alternative ridge estimator, referred to as euclidean distance ridge (\nameedr), and calibrate the tuning parameter using adaptive validation~\cite{lepskii_1992, spokoiny_mammen_lepski_1997} \textit{individually} for each patient.
We show that this approach offers compelling theory
, fast computations
, and accurate prediction on data. 

The specific motivation for our method is to unravel the relationship between gene expression and weight gain in kidney transplant recipients~\cite{cashion2013expression}.
Kidney transplant recipients are known to often gain substantial weight during the first year after transplantation, which can result in adverse health effects~\cite{patel1998effect}. 
Individual predictions of this weight gain based on the genetic data could help in providing each patient with the best possible care.

 


The remainder of this paper is organized as follows:
We introduce the linear regression framework and the problem statement in Section~\ref{sec:ProblemSetup}. 
We then introduce the main methodology of our approach, and present theoretical guarantees in Section~\ref{sec:Methodology}. 
In addition, we discuss the algorithm and analyze its performance through simulation studies using synthetic and real data in Section~\ref{sec:Algorithm}. 
We further apply our pipeline to kidney transplant data in Section~\ref{sec:Application}.
Finally, we discuss the results in Section~\ref{sec:Conclusion} and we defer all proofs to the Appendix.
All data are publicly available and our code is available at \url{https://github.com/LedererLab/personalized_medicine}.

\section{Problem Setup}
\label{sec:ProblemSetup}

We consider data ($\outcome, \design$) that follows a linear regression model 
\begin{equation}
\label{eq:LinearModel}
    \outcome = \design\true + \noise.
\end{equation}
Let $p$ denote the number of parameters, e.g.\ genes or genetic probes, and $n$ the number of samples or patients, then $\outcome \in \real{n}$ is the vector of outcomes, $y_i$, for example, a person’s response to treatment. We let~$\design$ denote the design matrix, where each row $\mathbf{x}_i\in\real{p}$, $i\in\{1,\dots n\}$, contains the genome information of the corresponding person.
%
Each element~$\beta^*_j$, $j\in\{1,\dots p\}$, of the regression vector~$\true\in\real{p}$ models the gene's influence on the person's response. 
We ensure the uniqueness of~$\true$ by assuming that it is a projection onto the linear space generated by the~$n$ rows of~$\design$~\cite{Shao2012,Buhlmann2013}.
For the random error vector ~$\noise \in \real{n}$, we make no assumptions on the probability distribution.

Our goal is to estimate the regression vector~$\true$ from data~$(\outcome,\design)$, or in terms of our application, predicting a person's treatment response based on that person's genome information. Mathematically, this amounts to estimating $\testing^\top\true$ in terms of the personalized prediction error
\begin{equation}
\label{eq:SingleError}
    \bigl|\testing^{\top}(\true-\estimator)\bigr|,
\end{equation} 
where $\testing \in \real{p}$ is the person's genome information.

Since the data in precision medicine is typically high-dimensional, that is, the number of parameters (genes)~$p$ exceeds the number of samples (patients)~$n$, 
we consider regularized least-squares estimators of the form 
\begin{equation}
\label{eq:RegularizedLeastSquare}
    \regularizedestimator{}{\tuning} \in \argmin_{\target \in \real{p}} \bigg\{ \normtwo{\outcome - \design\target}^{2} + \tuning \cdot f[\target] \bigg\}.
\end{equation}
Here, $f$ denotes a  function that takes into account prior information, such as sparsity or smaller regression coefficients,
and the tuning parameter~$r \geq 0$ balances the least-squares term and the prior term.

Given an estimator~\eqref{eq:RegularizedLeastSquare},
the main challenge is to find a good tuning parameter in line with our statistical goal. 
This means that we want to mimic the tuning parameter
\begin{equation*}
\label{eq:OptimalTunningParameter}
    \optimal := \argmin_{\tuning \in \tuningset} \Bigl| \testing^{\top}(\true-\regularizedestimator{}{\tuning})\Bigr|,
\end{equation*}
which is the optimal tuning parameter in a given set of candidate parameters $\mathcal{R} :=\{ \tuning_{1}, \tuning_{2},\dots ,\tuning_{m}\}$.

The optimal tuning parameter~$\optimal$ depends on the family of estimators~\eqref{eq:RegularizedLeastSquare}, the unknown noise~$\noise$, and the patient's genome information~$\testing$.
The dependence on~$\testing$ is integral to personalized medicine:
different patients can respond very differently to the same treatment. 
But standard tuning parameter calibration such as CV schemes do not take this personalization into account but instead 
attempt to minimize the averaged prediction error $\normtwo{\design\true-\design\regularizedestimator{}{\tuning}}^2/n$ rather than the personalized prediction error $\abs{\testing^{\top}(\true-\regularizedestimator{}{\tuning})}$.
We, therefore, develop a new prediction pipeline, that is tailored to the personalized prediction error and equip our methods with fast algorithms and sharp guarantees.


\section{Methodology}
\label{sec:Methodology}

In this section, we introduce an alternative version of the ridge estimator~\cite{hoerl1970ridge} along with a calibration scheme tailored to personalized medicine.
Two distinct features of the pipeline are its finite-sample bounds and its computational efficiency.
Our estimator is called \emph{euclidean distance ridge} (\emph{$\nameedr$}) and is defined as 
\begin{equation}
\label{eq:SPR}
    \regularizedestimator{\nameedr}{\tuning} \: \in \: \argmin_{\target \in \real{p}} \bigg\{\normtwo{\outcome - \design\target}^{2} + \tuning \normtwo{\target}\bigg\}.
\end{equation}

The \nameedr\ replaces the ridge estimator's squared $\ell_2$ prior term $f_{\operatorname{ridge}}[\target]\equiv\normtwos{\target}$ by its square-root $f_{\nameedr}[\target]\equiv\sqrt{f_{\operatorname{ridge}}[\target]}\equiv\normtwo{\target}$.
This modification allows us to derive finite-sample oracle inequalities that can be leveraged for tuning parameter calibration.
At the same time, the  $\nameedr$  preserves two of the ridge estimator's most attractive features: 
it can model the influences of many parameters, and it can be computed without the need for elaborate descent algorithms (see Section~\ref{sec:Algorithm}).

Our first step is to establish finite-sample guarantees for the $\nameedr$.
The key idea is that if the tuning parameter is large enough, the personalized prediction error~\eqref{eq:SingleError} is bounded by a multiple of the tuning parameter.
For ease of presentation, we assume an orthonormal design, that is,  $\design^{\top}\design = ~\identity_{p \times p}$ and defer the discussion of correlated covariates to the Appendix~\ref{appendix:orthonormal}.
However, simulations with more general designs are carried out in Section~\ref{sec:Algorithm}.
We establish the following guarantee for \nameedr:
\begin{lemma}[Oracle inequality for $\nameedr$]
\label{lem:OracleSPR}
If $\tuning\geq 2\abs{(\design\testing)^\top\noise}/(\corr{\testing}{\tuning}\normtwo{\testing})$, where
\begin{equation*}
    \corr{\testing}{\tuning} :=  \frac{\abs{\testing^{\top}\regularizedestimator{\nameedr}{\tuning}}}{\normtwo{\testing}\normtwo{\regularizedestimator{\nameedr}{\tuning}}} \quad\in[0,1],
\end{equation*}
then it holds for orthonormal design that 
\begin{equation*}
    \bigl|\testing^{\top}(\true - \regularizedestimator{\nameedr}{\tuning})\bigr| \leq \corr{\testing}{\tuning} \cdot \normtwo{\testing} \:\cdot\: \tuning.
\end{equation*}
    
\end{lemma}
\noindent Such guarantees are usually called \emph{oracle inequalities}~\cite{lederer2019oracle}. 
The given oracle inequality is an ideal starting point for our pipeline, because it gives us a mathematical handle on the quality of tuning parameters:
a good tuning parameter should be large enough to meet the stated condition and yet small enough to give a sharp bound.
The original ridge estimator, however, lacks such inequalities for personalized prediction.
 
Our proof techniques, which are based on the optimality conditions of the estimator, also yield a similar bound for the original ridge estimator:
if $\ridgetuning\geq \abs{(\design\testing)^\top\noise}/\normtwo{\testing}$,
then $\abs{\testing^{\top}(\true - \regularizedestimator{\nameridge}{\ridgetuning})}\leq \abs{ 1 + \testing^{\top}\regularizedestimator{\nameridge}{\ridgetuning} / \normtwo{\testing}} \cdot \normtwo{\testing} \cdot \ridgetuning$.
The following pipeline can then be applied the same way as for the \nameedr.
But the crucial advantage of the \nameedr's bound is that its right-hand side is bounded by $\normtwo{\testing} \:\cdot\: \tuning$,
which ensures that the results do not scale with~\true.

The factor $\corr{\testing}{\tuning}$ can be interpreted as the absolute value of the correlation between the person's genome information~$\testing$ and the estimator~$\regularizedestimator{\nameedr}{\tuning}$. 
This factor, and therefore~\testing, are included in our calibration scheme below, and our pipeline, hence, optimizes the prediction for particular study subjects. 

Lemma~\ref{lem:OracleSPR} bounds the personalized prediction error of $\nameedr$ as a function of the tuning parameter~\tuning.
Given~\testing, the best tuning parameter in terms of the bound minimizes $\corr{\testing}{\tuning}\cdot\tuning$ over all tuning parameters, that satisfy the lower bound  
\begin{equation*}
\tuning\geq \frac{2\abs{(\design\testing)^\top\noise}}{\corr{z}{r}\normtwo{\testing}}.
\end{equation*}
This tuning parameter value, which we call the oracle tuning parameter, can be interpreted as the closest theoretical mimic of the optimal tuning parameter~\optimal.

\begin{definition}[Oracle tuning parameter for personalized prediction]
\label{def:OracleTuning}
Given a new person's genome information \testing, the oracle tuning parameter for personalized prediction in a candidate set \tuningset\ is given by
    \begin{equation*}
        \oracletuning \: \in \: \argmin_{\tuning\in \bar{\tuningset}} \biggl\{ \corr{\testing}{\tuning} \cdot\tuning  \biggr\}, 
        \text{ where }
        \overline{\tuningset} := \biggl\{ \tuning \in \tuningset : \tuning\geq\frac{2\abs{(\design\testing)^\top\noise}}{\corr{z}{r}\normtwo{\testing}}\biggr\}.
    \end{equation*}
\end{definition}
\noindent 
The oracle tuning parameter~\oracletuning\ is the best approximation of the optimal tuning parameter~$\optimal$ in view of the mathematical theory expressed by Lemma~\ref{lem:OracleSPR}.
In practice, however, one does not know the target~\true\ nor the noise~$\noise$ (typically not even its distribution), such that neither~\optimal\ nor \oracletuning\ are accessible. 

In the following our goal is, consequently, to match the prediction accuracy of~\oracletuning\ (and, therefore, of~\optimal\ essentially) with a completely data-driven scheme.
Our proposal is based on pairwise tests along the tuning parameter path:
\begin{definition}[$\pavedr$ : Personalized adaptive validation for $\nameedr$]
\label{def:AV}
  We select a tuning parameter~\avtuning\ by
    \begin{equation}
    \label{eq:PAVedrTuning}
        \avtuning \: \in \: \argmin_{\tuning\in\avtuningset} \bigg\{ \corr{\testing}{\tuning} \cdot\tuning \cdot \normtwo{\testing} \bigg\},
    \end{equation}
    where the set of admissible tuning parameters is
    \begin{equation*}
    \begin{split}
        \avtuningset := \Bigg\{ \tuning \in \tuningset \, \bigg| \, \max_{\substack{\tuningone,\tuningtwo\in\tuningset \\ \tuningone,\tuningtwo\geq\tuning}} \bigg[ & \abs{\testing^{\top}(\regularizedestimator{\nameedr}{\tuningone} - \regularizedestimator{\nameedr}{\tuningtwo})} \\
        &- (\corr{\testing}{\tuningone} \cdot\tuningone + \corr{\testing}{\tuningtwo} \cdot\tuningtwo)\normtwo{\testing} \: \leq 0 \bigg] \Bigg\}.
    \end{split} 
    \end{equation*}
\end{definition} 
\noindent The idea of using pairwise tests for tuning parameter calibration in high-dimensional statistics has been introduced by~\cite{Chichignoud2016} under the name \emph{adaptive validation}.
A difference here is that the factors $\corr{\testing}{\tuning}\cdot\tuning$ are not constant but depend both on~\tuning\ and~\testing. 
The dependence on~\testing\ in particular reflects our focus on \emph{personalized} prediction.

The following result guarantees that the data-driven choice $\avtuning$ indeed provides---up to a constant factor 3---the same performance as the oracle tuning parameter~\oracletuning.
\begin{thm}[Optimality for personalized  adaptive validation for \nameedr]
\label{thm:optimality}
Under the conditions of Lemma~\ref{lem:OracleSPR}, it holds that 
     \begin{equation*}
         \abs{\testing^{\top}(\true-\regularizedestimator{\nameedr}{\avtuning})} \leq 3 \: \corr{\testing}{\oracletuning}\cdot \normtwo{\testing} \:\cdot\: \oracletuning.
     \end{equation*}
\end{thm}
\noindent This result guarantees that our calibration pipeline selects an essentially optimal tuning parameter from any grid~\tuningset.
Our pipeline is the only method for tuning parameter selection in personalized medicine that is equipped with such finite-sample guarantees. 
It does, moreover, not require  any knowledge about the regression vector~\true\ nor the noise~$\noise$.

Our calibration method is fully adaptive to the noise distribution; 
however, it is instructive to exemplify our main result by considering Gaussian noise (see  Appendix~\ref{example:gaussian} for the detailed derivations):
\begin{example}[Gaussian noise]
\label{ex:OracleNormalSPR}
Suppose orthonormal design and Gaussian random noise $\noise\sim \mathcal{N}_{n}[0_{n},\sigma^{2}\identity_{n \times n}/n]$. For any $\errorprobability \in (0,1)$, 
it holds with probability at least $1-\errorprobability$ that
    \begin{equation*}
        \bigl|\testing^{\top}(\true-\regularizedestimator{\nameedr}{\avtuning})\bigr| \leq 3 \sigma  \sqrt{\frac{8\log({2}/{\errorprobability})}{n}}\normtwo{\testing}.
    \end{equation*}
The bound provides the usual parametric rate~$\sigma/\sqrt{n} $ in the number of samples~$n$;
the factor~$\normtwo{\testing}$ entails the dependence on the number of parameters~$p$.
\end{example}


\section{Algorithm and Numerical Analysis}
\label{sec:Algorithm}

One of the main features of our pipeline is its efficient implementation.
This implementation exploits a fundamental property of our estimator:
there is a one-to-one correspondence between the \nameedr\ and the ridge estimator via the tuning parameters.


\subsection{Connections to the ridge estimator}

The ridge estimator is the $\ell_2^2$-regularized least-squares estimator~\cite{hoerl1970ridge}
\begin{equation}
\label{eq:RidgeRegression}
  \regularizedestimator{\nameridge}{\ridgetuning} \in \argmin_{\target \in \real{p}} \biggl\{\normtwo{\outcome-\design\target}^{2} + \ridgetuning  \normtwo{\target}^{2}\biggr\},
\end{equation}
where $\ridgetuning>0$ is a tuning parameter. Its computational efficiency, which is due to its closed-form expression, provides a basis for the computation of our \nameedr\ estimator.
The closed-form of the ridge estimator can be derived from the Karush-Kuhn-Tucker (KKT) conditions as
\begin{equation}
\label{eq:RidgeEstimator}
    \regularizedestimator{\nameridge}{\ridgetuning}=(\design^{\top}\design+\ridgetuning\identity_{p \times p})^{-1}\design^{\top}\outcome, 
\end{equation}
noting that the matrix $(\design^{\top}\design+\ridgetuning\identity_{p \times p})$ is always invertible if  $\ridgetuning>0$. 

However, the inversion of the matrix $\design^{\top}\design+\ridgetuning\identity_{p \times p}$ still deserves some thought:
first, the matrix might be ill-conditioned,
and second, the matrix needs to be computed for a range of tuning parameters rather than only for a single one.
A standard approach to these two challenges is a  singular value decomposition (svd) of the design matrix~\design.

\begin{lemma}[Computation of the ridge estimator through singular value decomposition]
\label{lem:SVD}
Let a singular value decomposition of \design\ be given by $\design = \leftsingular \diagnal\rightsingular^{\top}$, where 
$\leftsingular\in\real{n \times n}$ and $\rightsingular\in\real{p \times p}$ are orthonormal matrices, and 
$\diagnal = \operatorname{diag}(\singularvalue_{1}, \singularvalue_{2}, . . . , \singularvalue_{p})$ is an $n \times p$ diagonal matrix of the corresponding {\em singular values} $d_{1}, d_{2}, . . .  , d_{p}$. 
Then, the ridge estimator can be computed as
\begin{equation}
\label{eq:RidgeBySVD}
    \regularizedestimator{\nameridge}{\ridgetuning}
     = \rightsingular\diagnalinv\leftsingular^{\top}\outcome, 
\end{equation}
where $\diagnalinv \in \real{p \times n}$ is diagonal with $\diagnalinv = \operatorname{diag}(\singularvalue_{1}/(\singularvalue_{1}^{2}+\ridgetuning), ... , \singularvalue_{p}/(\singularvalue_{p}^{2}+\ridgetuning))$.
\end{lemma}
\noindent The singular value decomposition of the design matrix does not depend on the tuning parameter; 
therefore, the ridge estimators \regularizedestimator{\nameridge}{\ridgetuning}\ can be readily computed for multiple tuning parameters just by substituting the value of \ridgetuning\ in~\diagnalinv. 
The resulting set of ridge ($\nameedr$) estimators for a set of tuning parameters \ridgetuningset\ is called the ridge ($\nameedr$) path for \ridgetuningset.

Now, the crucial result is that the ridge estimator and the \nameedr\ are computational siblings.
\begin{thm}[One-to-one mapping between tuning parameters]
\label{thm:TuningRelation}
The one-to-one mapping $\tuningmapping{\ridgetuning}\, : \, \ridgetuning \mapsto \tuning$ defined by 
\begin{equation}
\label{eq:TuningMapping}
    \tuning = \tuningmapping{\ridgetuning} : = \normtwo{2\design^{\top}(\outcome-\design\regularizedestimator{\nameridge}{\ridgetuning})}
\end{equation}
transforms tuning parameters~\ridgetuning\ of the ridge estimator to tuning parameters~\tuning\ of the \nameedr\ estimator such that $\regularizedestimator{\nameridge}{\ridgetuning}=\regularizedestimator{\nameedr}{\tuning}$.
\end{thm}

\noindent This mapping transforms, in particular, the optimal tuning parameter of the ridge\ estimator to a corresponding optimal tuning parameter of the \nameedr\ estimator. 
More generally, it allows us to compute the $\nameedr$ estimator via the ridge estimator---see below.


\subsection{Algorithm}
\label{subsec:Algorithm}

The core idea of our proposed algorithm is to exploit the above one-to-one mapping between  \nameedr\ estimator and ridge\ estimator.
This correspondence allows us to compute $\nameedr$ solution paths efficiently via the ridge's explicit formulation and svd. 

First, consider a set of ridge\ tuning parameters \ridgetuningset\ and its corresponding set of \nameedr\ tuning parameters given by
\begin{equation*}
\label{eq:TransformedTuningSet}
\ridgetuningsetmapped:=\biggl\{\tuning\in\real{}:\>r=\tuningmapping{\ridgetuning},\,\ridgetuning\in\ridgetuningset\biggr\}
\end{equation*}
with cardinality $m:=|\ridgetuningsetmapped|$.
This set contains, in particular, the tuning parameter~$\avtuning$, 
whose optimality is guaranteed under Theorem~\ref{thm:optimality}.
To compute the tuning parameter~$\avtuning$, given data~\testing, we first order the elements $r_1,r_2,\dots,r_m$ of~\ridgetuningsetmapped\ such that
\begin{equation}
\label{eq:TuningSorting}
    \corr{\testing}{\tuning_{1}} \cdot\tuning_{1} \le \corr{\testing}{\tuning_2}\cdot\tuning_2 \le \dots \le
     \corr{\testing}{\tuning_{m}} \cdot\tuning_{m}.
\end{equation}
The $\pavedr$ method can then be formulated in terms of the binary random variables
\begin{equation*}
    \bitesting_{\tuning_{i}} :=  \prod_{j=i}^{m}\mathbbm{1} \Bigg\{ 
     \abs{\testing^{\top}(\regularizedestimator{\nameedr}{\tuning_{i}} - \regularizedestimator{\nameedr}{\tuning_{j}})} 
    - \Bigl(\corr{\testing}{\tuning_{i}} \cdot\tuning_{i} + \corr{\testing}{\tuning_{j}} \cdot\tuning_{j}\Bigr)\normtwo{\testing} \: \leq \, 0 \Bigg\}
\end{equation*}
for $i \in \{1,\dots,m\}$, 
and an algorithm is as follows:

\begin{algorithm}[H]
\label{alg:PAVEDR}
\SetAlgoLined

\textbf{Input}: $\bigl(\tuning_{i}\bigr)_{i=1,\dots,m}, 
\bigl(\regularizedestimator{\nameedr}{\tuning_{i}}\bigr)_{i=1,\dots,m}, 
\testing $

\textbf{Result}: $\avtuning$

\vspace{1ex}
Set initial index: $i\gets m$

 \While{$\bitesting_{\tuning_{i}} \neq 0$ and $i > 1$}{
  Update index: $i\gets i-1$
 }

Set output: $\avtuning \gets \tuning_{i}$
\vspace{1em}

\caption{Algorithm for $\pavedr$ of Definition~\ref{def:AV}.}
\end{algorithm}%
\noindent The full pipeline can be summarized by the following four steps:
\begin{enumerate}[leftmargin=1.8cm]

    \item[\textit{Step 1:}] Generate a set \ridgetuningset\ of tuning parameters for ridge regression.
    \item[\textit{Step 2:}] Compute the ridge solution path with respect to \ridgetuningset\ by using~\eqref{eq:RidgeBySVD}.
    \item[\textit{Step 3:}] Transform the ridge tuning parameters to their \nameedr\ counterparts \ridgetuningsetmapped\ using~\eqref{eq:TuningMapping} and sort the tuning parameters according to~\eqref{eq:TuningSorting}.
    \item[\textit{Step 4:}] Use the \pavedr\ method (Algorithm~\ref{alg:PAVEDR}) to compute the tuning parameter~\avtuning\ and map it back to its ridge counterpart \avtuningridge.
    
\end{enumerate}

The algorithm can be readily implemented and is fast:
it essentially only requires the computation of one ridge solution path (a single svd). 
In strong contrast, $K$-fold CV requires the computation of $K$ ridge solution paths. 
Consequently, the ridge estimator with \pavedr\ can be computed approximately $K$~times faster than with $K$-fold CV, which we will confirm in the simulations.
Moreover, CV still requires a tuning parameter, namely, the number of folds~$K$, while \pavedr\ is completely parameter-free.


\subsection{Simulation Study}
\label{subsec:SimulationStudy}

We evaluate the prediction performance of the \pavedr\ method using (1) fully simulated data with random design and (2) a real data set with a simulated outcome.
The results are compared to $K$-fold CV, which is a standard reference method. 

The first setting is solely based on simulated data. The dimensions of the design matrix are $(n,p) \in \{ (50,100), (150,250), (200,500) \}$. 
First, the entries of the design matrix~$\design$ are sampled i.i.d.\@ from~$\normalDist{\mu}{1}$, where the mean itself is sampled according to $\mu\sim\normalDist{0}{10}$,
and the columns of the design matrix are then normalized to have Euclidean norm equal to one.
The entries of the regression vector~\true\ are sample i.i.d.\@ from $\normalDist{0}{1}$ and then projected onto the row space of \design\ to ensure identifiability~\cite{Shao2012,Buhlmann2013}.
The entries of the noise vector~$\noise$ are sampled i.i.d.\@ from $\normalDist{0}{\sigma^2}$,  where  $\sigma^2 = 2 \variance{\design\true}$ to ensure a signal-to-noise ratio of 0.5.
Then, $100$~data testing vectors~$\testing$ are sampled i.i.d.\@ from $\mathcal{U}_p[-1,1]$.
We generate a set of 300 tuning parameters $\ridgetuningset = \{10^q \>|\> q=-5 + 10i/299, \ i=0,\dots,299 \}$.


\begin{table}
\centering
\caption{For the first simulation setting, which entirely consists of artificial data, \pavedr\ outperforms $5$-fold and $10$-fold CV in accuracy and speed.}
\label{table:simulationstudyRandom}
\ra{1.2}
\begin{tabular}{ccrr}
    \hline
    (n,p) & Method & Mean error (sd) & Scaled run time\\ 
    \hline
    \multirow{3}{*}{(50,100)}& $\pavedr$ & 166.78 \hphantom{0}(242.46) & 1.00\\
                             & 5-fold CV & 340.18 \hphantom{0}(888.28) & 1.57\\
                             & 10-fold CV & 474.58 (1220.44) & 3.64\\
    \hline
    \multirow{3}{*}{(150,250)}& $\pavedr$ & 433.50 \hphantom{0}(669.50) & 1.00\\
                             & 5-fold CV & 724.90 (1712.65) & 3.43\\
                             & 10-fold CV & 872.50 (2560.01) & 8.04\\
    \hline
    \multirow{3}{*}{(200,500)}& $\pavedr$ & 805.94 (1316.43) & 1.00\\
                             & 5-fold CV & 1098.68 (2821.35) & 3.65\\
                             & 10-fold CV & 1144.12 (2733.78) & 8.44\\
    \hline
\end{tabular}
\end{table}

The results are summarized in Table~\ref{table:simulationstudyRandom}.
The mean personalized prediction errors for the testing vectors are averaged over 100 simulations as described above.
The run time is shown relative to \pavedr.
We observe that in all considered cases, \pavedr\ improves on CV both in terms of accuracy as well as in speed. A more detailed analysis of the scaled run time for CV relative to $\pavedr$ is shown in Figure~\ref{fig:runtime}. We fix $n$ with increasing $p$ and vice versa. Observing that the gain in speed is less than the factor~$K$, because the  computations of the ridge estimator are then fast enough to compete with the sorting of the bounds  in \pavedr.


\begin{figure}
    \centering
    \includegraphics[width=0.49\textwidth]{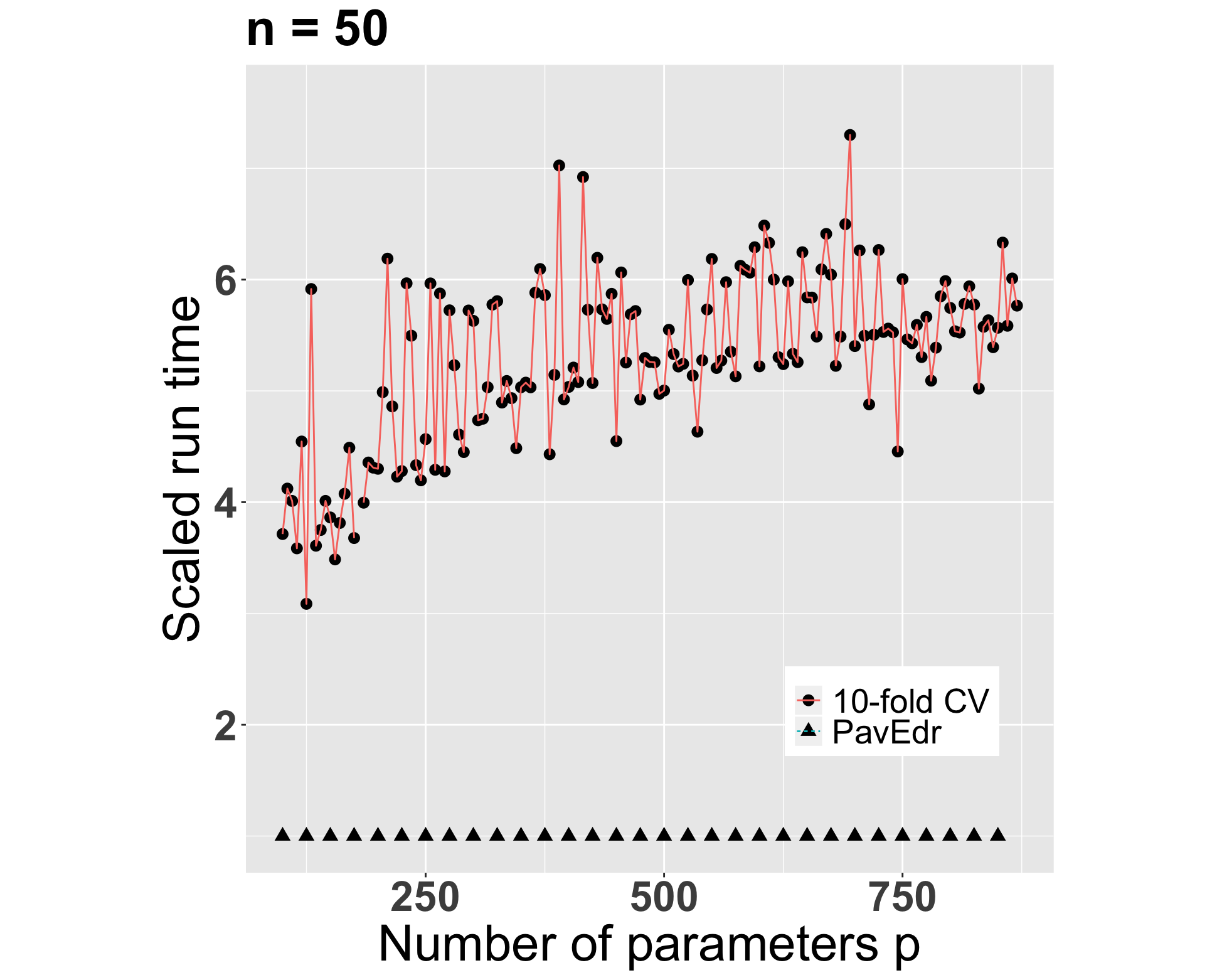} 
    \hfill
    \includegraphics[width=0.49\textwidth]{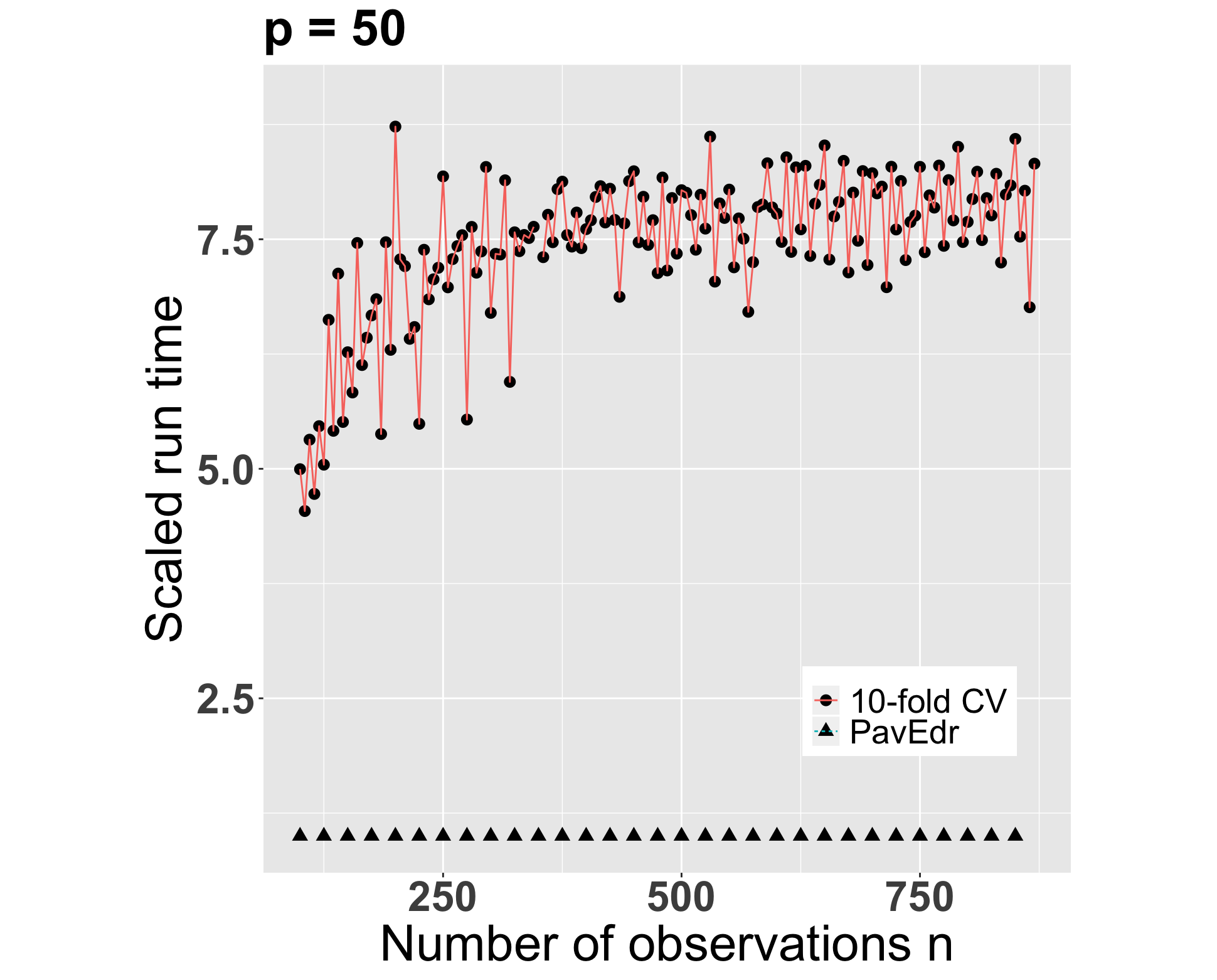}
\caption[]{Run time of 10-fold CV scaled by the run time of \pavedr\ for a fixed number of observations $n$ and fixed number of parameters $p$ with $p$ and $n$ increasing, respectively. We observe that $\pavedr$ is faster than $10$-fold CV.}
\label{fig:runtime}
\end{figure}

In the second setting, we simulate the outcome but based on real data as covariates. The basis is the genomic data from the application in Section~\ref{sec:Application} where the sample size or number of patient is~$n=26$. 
The number of covariates in the design matrix is restricted to the $p=1936$ gene probe targets 
identified as potentially influential by \cite{cashion2013expression}. 
The regression vector and the noise are then generated as in the first simulation setting above.  The results are summarized in Table~\ref{table:simulationstudyKidney}.
We observe again that \pavedr\ improves on CV both in terms of accuracy as well as in speed. The results of this section demonstrate that \pavedr\ is a contender on data,
which confirms and complements our theoretical findings from before.

\begin{table}
\centering
\caption{In the second simulation setting, which consists of real covariate data and simulated outcomes, \pavedr\ outperforms  $5$-fold and $10$-fold CV again both in accuracy and speed.}
\label{table:simulationstudyKidney}
\ra{1.2}
\begin{tabular}{crr}
    \hline
    Method & Mean error (sd) & Scaled run time\\ 
    \hline
    $\pavedr$  & \hphantom{0}33.71 \hphantom{0}(34.73) & 1.00\\
    5-fold CV  & 164.06 (120.06) & 4.02\\
    10-fold CV & 132.90 \hphantom{0}(97.31) & 9.58\\
        
    \hline
\end{tabular}
\end{table}


\section{Application to kidney transplant patient data}
\label{sec:Application}

Kidney transplant recipients are known to gain significant weight during the first year after transplantation, with a reported average increase of 12 kg~\cite{patel1998effect}. 
Such substantial weight gain over a relatively short time period gives an increased risk for several adverse health effects, such as cardiovascular disease, 
and may be detrimental for the overall outcome of the patient. The weight gain has been explained by the use of prescribed steroids which increase the appetite, 
but steroid-free protocols alone have not reduced the risk of obesity, suggesting alternative causes. 
Even though weight gain is fundamentally caused by a too high calorie intake relative to the energy expenditure, the heterogeneity in the individual response is substantial.
Genetic variation has, therefore, been considered as a contributing factor, and several genes have been linked to obesity and weight gain~\cite{bauer2009obesity, cheung2010obesity}. 

\cite{cashion2013expression} investigated whether genomic data can be used to predict weight gain in kidney transplant recipients. This was done by measuring gene expression in subcutaneous adipose tissue which has an important role in appetite regulation and can easily be obtained from the patients during surgery. The patients' weight was recorded at the time of transplantation and at a 6-months follow-up visit, resulting in a relative weight difference. The adipose tissue samples were collected from 26 transplant patients at the time of surgery, and mRNA levels were measured to obtain the gene expression profiles for $28\,869$ gene probe targets using Affymetrix Human Gene 1.0 ST arrays\footnote{All data are publicly available in the EMBL-EBI ArrayExpress database (\url{www.ebi.ac.uk/arrayexpress}) under accession number E-GEOD-33070.}. 
The expression variability was further not associated with gender or race~\cite{cashion2013expression}. 
As excessive weight gain may have severe consequences for the patients, the goal is to predict the future weight increase based on the available gene expression profiles. If a large weight increase is predicted, additional measures such as diet restrictions or physiotherapy could be set into effect. 

\begin{table} 
\centering
\caption{
In the kidney transplant data, regardless of in-sample or leave-one-out prediction,  \pavedr\ outperforms  $5$-fold and $10$-fold CV again both in accuracy and speed. 
}
\label{sx9x9s}  
\ra{1.2}
\begin{subtable}{\textwidth}
\centering
\ra{1.2}
\caption{In-sample prediction}
\label{table:RealKidneyInSample}
\begin{tabular}{crr}
    \hline
    Method & Mean error (sd) & Scaled run time\\ 
    \hline
    $\pavedr$ & 0.0049 (0.0121) & 1.00\\
    5-fold CV & 0.0646 (0.0540) & 1.19\\
    10-fold CV & 0.0666 (0.0584) & 3.04\\
    \hline
\end{tabular}
\end{subtable}%
\hfill
\begin{subtable}{\textwidth}
\centering
\caption{Leave-one-out prediction}
\label{table:RealKidneyLoocv}
\begin{tabular}{crr}
    \hline
    Method & Mean error (sd) & Scaled run time\\ 
    \hline
    $\pavedr$ & 0.0622 (0.0309) & 1.00\\
    5-fold CV & 0.0672 (0.0415) & 1.10\\
    10-fold CV & 0.0678 (0.0475) & 2.82\\
    \hline
\end{tabular}
\end{subtable}
\end{table}

We compare the performance of our method in predicting weight gain for the kidney transplant patients to the prediction of standard ridge regression calibrated by CV. 
In detail, we make predictions for each patient both in-sample and out-of-sample, leaving out the observation and using the remaining data to fit the penalized regression model and select the optimal tuning parameter. 
Since we do not know the true parameter $\boldsymbol{\beta^{*}}$, we can only examine the performance of our method and CV by comparing their estimation errors, which is defined by 
\begin{equation}
\label{eq:KidneyError}
|y_{i} - \mathbf{x}_i^{\top}\regularizedestimator{\nameedr}{\tuning}|.
\end{equation}
\noindent As described in the previous section, the  columns  of  the  design  matrix are normalized to have Euclidean norm one. Unlike in the Section \ref{subsec:SimulationStudy}, we here take all the $28\,869$ gene probes into consideration. 

%

The averaged results are summarized in Table~\ref{table:RealKidneyInSample} and Table~\ref{table:RealKidneyLoocv}. 
We observe that $\pavedr$ clearly outperforms  5-fold and 10-fold CV for both in-sample and out-of-sample prediction of the kidney transplant data. 
For out-of-sample prediction, we observe an improvement of about $7.5\%$ in the estimation error and an improvement of $25.5\%$ in the standard deviation compared to 5-fold CV.
These improvements, especially in standard deviation, reinforce the advantages of a personalized approach to tuning parameter calibration.


\section{Conclusion}
\label{sec:Conclusion}

We have introduced a pipeline that calibrates ridge regression for personalized prediction. 
Its distinctive features are the finite sample guarantees (see Theorem~\ref{thm:optimality}) and the statistical and  computational efficiency (see Tables~\ref{table:simulationstudyRandom} and~\ref{table:simulationstudyKidney}).
These features are echoed when predicting the weight gain of kidney transplant patients (see Table~\ref{sx9x9s}).
Hence, our pipeline can improve personalized prediction and, thereby, further the cause of personalized medicine. 

Despite our focus on personalized medicine, we also envision applications in other areas where individual heterogeneity is crucial for predictions. 
Two examples are item recommendation, predicting the rating of an item or product assigned by a specific user~\cite{guy2010social,rafailidis2014content}, and personalized marketing, delivering individualized product prices or messages to specific costumers~\cite{tang2013prediction}.

\appendix


\section{Proofs}

\subsection{Proof of Lemma~\ref{lem:OracleSPR}}
\label{subsec:OracleSPRProof}
\begin{proof}

Assume $\tuning \geq 2\abs{(\design\testing)^\top\noise}/(\corr{\testing}{\tuning}\normtwo{\testing})$ and orthonormal design $\design^\top\design=\identity_{p\times p}$. According to the KKT conditions of the $\nameedr$ estimator, we have 
\begin{equation*}
\begin{split}
    \tuning\frac{\regularizedestimator{\nameedr}{\tuning}}{\normtwo{\regularizedestimator{\nameedr}{\tuning}}} & =
    2\design^{\top}(\outcome-\design\regularizedestimator{\nameedr}{\tuning}) \\
    & = 2\design^{\top}(\design\true+\noise-\design\regularizedestimator{\nameedr}{\tuning}) \\
    & = 2\design^{\top}\design(\true-\regularizedestimator{\nameedr}{\tuning})+2\design^{\top}\noise.
\end{split}
\end{equation*}
Hence,
\begin{equation}
\label{eq:Bound1}
    \design^{\top}\design(\true-\regularizedestimator{\nameedr}{\tuning})= -\design^{\top}\noise+\frac{\tuning}{2}\frac{\regularizedestimator{\nameedr}{\tuning}}{\normtwo{\regularizedestimator{\nameedr}{\tuning}}}.
\end{equation}
Let $\testing \in \real{p}$ and multiply $\testing^{\top}$ from the left to obtain
\begin{equation*}
    \testing^{\top}(\true-\regularizedestimator{\nameedr}{\tuning})= -\testing^{\top}\design^{\top}\noise+\frac{\tuning}{2}\frac{\testing^{\top}\regularizedestimator{\nameedr}{\tuning}}{\normtwo{\regularizedestimator{\nameedr}{\tuning}}}
\end{equation*}
where we use the assumption of orthonormal design. By taking absolute value on both sides and applying the triangle inequality, we derive the following bound for the personalized prediction error \eqref{eq:SingleError}:
\begin{equation*}
\begin{split}
    \abs{\testing^{\top}(\true-\regularizedestimator{\nameedr}{\tuning})} 
    & \leq \abs{\testing^{\top}\design^{\top}\noise} + \frac{\tuning}{2}\bigl|\frac{\testing^{\top}\regularizedestimator{\nameedr}{\tuning}}{\normtwo{\regularizedestimator{\nameedr}{\tuning}}}\bigr| \\
    & \leq \frac{\tuning}{2}\corr{\testing}{\tuning}\normtwo{\testing} +  \frac{\tuning}{2}\bigl|\frac{\testing^{\top}\regularizedestimator{\nameedr}{\tuning}}{\normtwo{\testing}\normtwo{\regularizedestimator{\nameedr}{\tuning}}}\bigr|\normtwo{\testing} \\
    & = \corr{\testing}{\tuning}\cdot\tuning\cdot\normtwo{\testing},
\end{split}
\end{equation*}
since $\tuning \geq 2\abs{ (\design\testing)^{\top}\noise}/(\corr{\testing}{\tuning}\normtwo{\testing})$ by assumption. Finally, we obtain the bound
\begin{equation}
    \abs{\testing^{\top}(\true-\regularizedestimator{\nameedr}{\tuning} )} \leq \corr{\testing}{\tuning} \cdot \normtwo{\testing} \:\cdot\: \tuning,
\end{equation}
with
\begin{equation*}
    \corr{\testing}{\tuning} :=  \frac{\bigl|\testing^{\top}\regularizedestimator{\nameedr}{\tuning}\bigr|}{\normtwo{\testing}\normtwo{\regularizedestimator{\nameedr}{\tuning}}}.
\end{equation*}

\end{proof}


\subsection{Proof of Theorem~\ref{thm:optimality}}
\label{subsec:optimalityProof}

\begin{proof}

Let $\testing\in\real{p}$ and suppose that the linear regression model \eqref{eq:LinearModel} is under orthonormal design.

\textbf{Bound on $\corr{\testing}{\avtuning} \cdot \avtuning$:}
First, we show that $\corr{\testing}{\avtuning}  \cdot \avtuning  \leq \corr{\testing}{\oracletuning} \cdot \oracletuning$. 
Let
\begin{equation*}
    \corr{\testing}{\avtuning} \cdot \avtuning  \geq \corr{\testing}{\oracletuning}  \cdot \oracletuning,
\end{equation*}
 then by definition of $\avtuning$, there must exist two tuning parameters~$\tuningone, \tuningtwo$ with 
 \begin{align*}
     \tuningone & \geq 2\abs{ (\design\testing)^{\top}\noise}/(\corr{\testing}{\tuningone}\normtwo{\testing}),\\
     \tuningtwo & \geq 2\abs{ (\design\testing)^{\top}\noise}/(\corr{\testing}{\tuningtwo}\normtwo{\testing}),
 \end{align*} such that 
\begin{equation*}
        \abs{\testing^{\top}(\regularizedestimator{\nameedr}{\tuningone} - \regularizedestimator{\nameedr}{\tuningtwo})} 
        \geq 
        (\corr{\testing}{\tuningone} \cdot\tuningone + \corr{\testing}{\tuningtwo} \cdot\tuningtwo)\cdot \normtwo{\testing}.
\end{equation*}

However, by Lemma \eqref{lem:OracleSPR}, we have 
\begin{equation*}
    \abs{\testing^{\top}(\true-\regularizedestimator{\nameedr}{\tuningone} )} \leq \corr{\testing}{\tuningone}  \cdot \tuningone \cdot \normtwo{\testing}
\end{equation*}
and 
\begin{equation*}
    \abs{\testing^{\top}(\true-\regularizedestimator{\nameedr}{\tuningtwo} )} \leq \corr{\testing}{\tuningtwo}  \cdot \tuningtwo \cdot \normtwo{\testing}.
\end{equation*}
Applying the triangle inequality to the above displays and combining the results yields
\begin{equation*}
        \abs{\testing^{\top}(\regularizedestimator{\nameedr}{\tuningone} - \regularizedestimator{\nameedr}{\tuningtwo})} 
        \leq
        (\corr{\testing}{\tuningone} \cdot\tuningone + \corr{\testing}{\tuningtwo} \cdot\tuningtwo)\cdot \normtwo{\testing},
\end{equation*}
which leads to a contradiction to our assumption. Therefore, we obtain the following bound with respect to \oracletuning:
\begin{equation*}
       \corr{\testing}{\avtuning}  \cdot \avtuning  \leq \corr{\testing}{\oracletuning} \cdot \oracletuning.
\end{equation*}


\textbf{Bound on the personalized prediction error:}
Since $\corr{\testing}{\avtuning}  \cdot \avtuning  \leq \corr{\testing}{\oracletuning} \cdot \oracletuning$, we have
\begin{equation*}
\begin{split}
        \abs{\testing^{\top}(\regularizedestimator{\nameedr}{\avtuning} - \regularizedestimator{\nameedr}{\oracletuning})} 
        & \leq
        (\corr{\testing}{\avtuning} \cdot\avtuning + \corr{\testing}{\oracletuning} \cdot\oracletuning)\cdot \normtwo{\testing} \\
        & \leq
         2 \cdot \corr{\testing}{\oracletuning} \cdot\oracletuning\cdot \normtwo{\testing}
\end{split}
\end{equation*}
Applying the triangle inequality, we ultimately find the bound 
\begin{equation*}
\begin{split}
    \abs{\testing^{\top}(\true-\regularizedestimator{\nameedr}{\avtuning})} 
    & = \abs{\testing^{\top}(\true-\regularizedestimator{\nameedr}{\oracletuning} + \regularizedestimator{\nameedr}{\oracletuning} - \regularizedestimator{\nameedr}{\avtuning})} \\
    & \leq \abs{\testing^{\top}(\true - \regularizedestimator{\nameedr}{\oracletuning})\bigr| + \bigl|\testing^{\top}(\regularizedestimator{\nameedr}{\oracletuning} - \regularizedestimator{\nameedr}{\avtuning})}\\
    & \leq
         3 \cdot \corr{\testing}{\oracletuning} \cdot\oracletuning\cdot \normtwo{\testing}.
\end{split}
\end{equation*}

\end{proof}


\subsection{Proof of Example~\ref{ex:OracleNormalSPR}}
\label{example:gaussian}
\begin{lemma}[Deviation inequality]
\label{lem:deviationInequality}
    For any standard normal variable $V\sim \mathcal{N}_{1}[0,1]$, we have the following concentration bound
    \begin{equation*}
        \mathbb{P}\bigl\{ \abs{V} \geq x \bigr\} \leq 2\exp\bigl\{ \frac{-x^2}{2} \bigr\} \qquad\qquad (x > 0).
    \end{equation*}
\end{lemma}

\begin{proof}

\probability{V > x} = \probability{e^{\lambda V} > e^{\lambda x}}
for all $\lambda$. Now by Markov's inequality,

\begin{equation*}
\begin{split}
    \probability{e^{\lambda V} > e^{\lambda x}} 
    & \leq \frac{\mathbb{E}[e^{\lambda V}]}{e^{\lambda x}} \\
    & = e^{\frac{\lambda^{2}}{2}-\lambda x}
\end{split}
\end{equation*}
For $\lambda = x$, we have \probability{V > x} $\leq e^{\frac{-x^{2}}{2}}$. Since the standard normal distribution is symmetric about 0, we obtain the desired result.  
\end{proof}

\noindent Using this concentration bound, we derive the results of Example~\ref{ex:OracleNormalSPR}.

\begin{proof}

\noindent 
Given a $\testing \in \real{p}$, Gaussian noise $\noise\sim \mathcal{N}_{n}[0_{n},\sigma^{2}\identity_{n \times n}/n]$ with variance $\sigma^2$, and suppose that the linear regression model \eqref{eq:LinearModel} is under orthonormal design. We first show that $ \probability{ 2\abs{(\design\testing)^{\top}\noise}/(\corr{\testing}{\tuning}\normtwo{\testing}) \geq \gaussiantuning } \leq \errorprobability$ for 
\begin{equation*}
    \gaussiantuning := \frac{\sigma\normtwo{\design\testing}}{(\corr{\testing}{\tuning}\normtwo{\testing})}\sqrt{\frac{8\log(2/\errorprobability)}{n}}
\end{equation*}
using the concentration bound, Lemma~\ref{lem:deviationInequality}:
\begin{equation*}
\begin{split}
    \probability{2\abs{(\design\testing)^{\top}\noise}/(\corr{\testing}{\tuning}\normtwo{\testing}) \geq \gaussiantuning} 
    & = \probability{\frac{\abs{(\design\testing)^{\top}\noise}}{\frac{\sigma\normtwo{\design\testing}}{(\corr{\testing}{\tuning}\normtwo{\testing})}\sqrt{1/n}}   \geq \frac{(\corr{\testing}{\tuning}\normtwo{\testing})\gaussiantuning}{2\sigma\normtwo{\design\testing}\sqrt{1/n}}  }  \\
    & \leq 2\exp\bigl\{ \frac{-(\frac{\sigma\normtwo{\design\testing}\sqrt{\frac{8\log(2/\errorprobability)}{n}}}{2\sigma\normtwo{\design\testing}\sqrt{1/n}})^2}{2} \bigr\}\\
    & =
        2\exp\bigl\{ \log(\errorprobability/2) \bigr\} \\
    & = \errorprobability.
\end{split}
\end{equation*}
Hence, $\gaussiantuning \geq 2\abs{(\design\testing)^{\top}\noise}/\corr{\testing}{\oracletuning}\normtwo{\testing}$ holds with at least probability $1-\errorprobability$. By Theorem~\ref{thm:optimality}, we have with at least probability $1-\errorprobability$:
\begin{align*}
     \abs{\testing^{\top}(\true-\regularizedestimator{\nameedr}{\avtuning})} 
     & \leq
     3 \ \corr{\testing}{\oracletuning} \,\oracletuning\normtwo{\testing} & (\abs{\corr{\testing}{\oracletuning}} \leq 1) \\
     & =
     3 \ \corr{\testing}{\oracletuning} \frac{\sigma\normtwo{\design\testing}}{\corr{\testing}{\oracletuning}\normtwo{\testing}}\sqrt{\frac{8\log(2/\errorprobability)}{n}}\normtwo{\testing}
     \\
     & =
     3 \sigma\sqrt{\frac{8\log(2/\errorprobability)}{n}}\normtwo{\testing}.
     & (\text{orthonormal design})
\end{align*}
\end{proof}


\subsection{Proof of Lemma~\ref{lem:SVD}}

\label{proof:SVDLemma}

\begin{proof}

Let $\design = \leftsingular \diagnal\rightsingular^{\top}$ be a singular value decomposition of \design\ as given in Lemma~\ref{lem:SVD}. Then by algebraic manipulation of Equation~\eqref{eq:RidgeEstimator} the ridge\ estimator can be written as
\begin{align*}
        \regularizedestimator{\nameridge}{\ridgetuning} 
        & =(\design^{\top}\design + \ridgetuning\identity_{p \times p})^{-1}\design^{\top}\outcome\nonumber\\
        & = (\rightsingular\diagnal^{T}\leftsingular^T\leftsingular\diagnal\rightsingular^{\top} + \ridgetuning\identity_{p \times p})^{-1}\rightsingular\diagnal\leftsingular^{\top}\outcome \nonumber\\
        & = (\rightsingular\diagnal^{2}\rightsingular^{\top} + \ridgetuning\identity_{p \times p})^{-1}\rightsingular\diagnal\leftsingular^{\top}\outcome \nonumber\\
        & = \rightsingular(\diagnal^{2} + \ridgetuning\identity_{p \times p})^{-1}\rightsingular^{\top}\rightsingular\diagnal\leftsingular^{\top}\outcome \nonumber\\ 
        & = \rightsingular\diagnalinv\leftsingular^{\top}\outcome,
\end{align*}
where the matrix $\diagnalinv$ is defined as
\begin{equation*}
    \diagnalinv = \operatorname{diag}(\frac{\singularvalue_{1}}{\singularvalue_{1}^{2}+\ridgetuning}, ... , \frac{\singularvalue_{p}}{\singularvalue_{p}^{2}+\ridgetuning}).
\end{equation*}

\end{proof}


\subsection{Proof of Theorem~\ref{thm:TuningRelation}}
\label{proof:TuningRelation}

\begin{proof}

We consider the KKT-conditions of \eqref{eq:SPR} and replace the \nameedr\ estimator with the ridge\ estimator to obtain
\begin{equation*}
    \tuning\frac{\regularizedestimator{\nameridge}{\ridgetuning}}{\normtwo{\regularizedestimator{\nameridge}{\ridgetuning}}}=2\design^{\top}(\outcome-\design\regularizedestimator{\nameridge}{\ridgetuning}).
\end{equation*}
By taking the $\ell_{2}$-norm of both sides and with $\tuning>0$, we obtain 
\begin{equation*}
    \tuning=\normtwo{ 2\design^{\top}(\outcome-\design\regularizedestimator{\nameridge}{\ridgetuning})}.
\end{equation*}
Thus, we can transform the ridge tuning parameter~\ridgetuning\ to the $\nameedr$ tuning parameter~$\tuning$ with respect to the same estimator. 

Moreover, there is a one-to-one relationship between $\nameedr$ and ridge. The ridge estimator in \eqref{eq:RidgeEstimator} implies that 
\begin{equation*}
    (\design^{\top}\design+\ridgetuning\identity_{p \times p})\regularizedestimator{\nameridge}{\ridgetuning}=\design^{\top}\outcome,
\end{equation*}
and hence 
\begin{equation*}
    \ridgetuning\regularizedestimator{\nameridge}{\ridgetuning} = \design^{\top}(\outcome-\design\regularizedestimator{\nameridge}{\ridgetuning}).
\end{equation*}
Since 
\begin{equation*}
    \tuning=\normtwo{2\design^{\top}(\outcome-\design\regularizedestimator{\nameridge}{\ridgetuning})} = 2\ridgetuning\normtwo{\regularizedestimator{\nameridge}{\ridgetuning}},
\end{equation*}
we have 
\begin{equation*}
   \frac{\tuning}{2\normtwo{\regularizedestimator{\nameridge}{\ridgetuning}}}=\ridgetuning
\end{equation*}
and we finally conclude  that $\regularizedestimator{\nameridge}{\ridgetuning} = \regularizedestimator{\nameedr}{\tuning}$ when  $\frac{\tuning}{2\normtwo{\regularizedestimator{\nameridge}{\ridgetuning}}}=\ridgetuning. $
\end{proof}

\section{Beyond Orthogonality}
\label{appendix:orthonormal}
To avoid digression, we have restricted the theories in the main body of the paper to orthonormal design matrices.
However, there are straightforward extensions along established lines in high-dimensional theory.
In general, the influence of correlation on regularized estimation has been studied extensively---see, for example, \cite{DalalyanHebiriLederer} and~\cite{Hebiri2013} for the lasso case.
The most straightforward extension of our theories goes via the $\ell_\infty$-restricted eigenvalue introduced in~\cite{Chichignoud2016}.
This condition allows for design matrices, that satisfy $|\!|\design^\top\design\boldsymbol{\delta}|\!|_\infty\gtrsim |\!|\boldsymbol{\delta}|\!|_\infty$ for certain $\boldsymbol{\delta}$.
We omit the details; importantly, our simulations demonstrate that our method provides accurate prediction far beyond orthonormal design.
\section{Beyond Orthogonality}
\label{appendix:orthonormal}



\begin{thebibliography}{100}

\bibitem[Bauer et al., 2009]{bauer2009obesity}
Florianne Bauer, Clara~C Elbers, Roger~AH Adan, Ruth~JF Loos, N~Charlotte
  Onland-Moret, Diederick~E Grobbee, Jana~V van Vliet-Ostaptchouk, Cisca
  Wijmenga, and Yvonne~T van~der Schouw.
\newblock Obesity genes identified in genome-wide association studies are
  associated with adiposity measures and potentially with nutrient-specific
  food preference.
\newblock {\em Am. J. Clin. Nutr.}, 90(4):951--959, 2009.

\bibitem[B{\o}velstad et al., 2007]{bovelstad2007predicting}
Hege~M B{\o}velstad, St{\aa}le Nyg{\aa}rd, Hege~L St{\o}rvold, Magne Aldrin,
  {\O}rnulf Borgan, Arnoldo Frigessi, and Ole~Christian Lingj{\ae}rde.
\newblock Predicting survival from microarray data-a comparative study.
\newblock {\em Bioinformatics}, 23(16):2080--2087, 2007.


\bibitem[B{\"u}hlmann, 2013]{Buhlmann2013}
Peter B{\"u}hlmann.
\newblock Statistical significance in high-dimensional linear models.
\newblock {\em Bernoulli}, 19(4):1212--1242, 2013.

\bibitem[Cashion et al., 2013]{cashion2013expression}
Ann Cashion, Ansley Stanfill, Fridtjof Thomas, Lijing Xu, Thomas Sutter, James
  Eason, Mang Ensell, and Ramin Homayouni.
\newblock Expression levels of obesity-related genes are associated with weight
  change in kidney transplant recipients.
\newblock {\em PloS one}, 8(3):e59962, 2013.

\bibitem[Cheung et al., 2010]{cheung2010obesity}
Chloe~YY Cheung, Annette~WK Tso, Bernard~MY Cheung, Aimin Xu, KL~Ong, Carol~HY
  Fong, Nelson~MS Wat, Edward~D Janus, Pak~C Sham, and Karen~SL Lam.
\newblock Obesity susceptibility genetic variants identified from recent
  genome-wide association studies: implications in a chinese population.
\newblock {\em J. Clin. Endocr. Metab.}, 95(3):1395--1403, 2010.

\bibitem[Chichignound, Lederer and Wainwright (2016)]{Chichignoud2016}
Michael Chichignoud, Johannes Lederer, and Martin~J. Wainwright.
\newblock {A Practical Scheme and Fast Algorithm to Tune the Lasso With
  Optimality Guarantees}.
\newblock {\em J. Mach. Learn. Res.}, 17(231):1--20, 2016.

\bibitem[Cho, Jeon and Kim, 2012]{cho2012personalized}
Sang-Hoon Cho, Jongsu Jeon, and Seung~Il Kim.
\newblock Personalized medicine in breast cancer: a systematic review.
\newblock {\em J. Breast. Canc.}, 15(3):265--272, 2012.

\bibitem[Cule and De Iorio, 2013]{cule2013ridge}
Erika Cule and Maria De~Iorio.
\newblock Ridge regression in prediction problems: automatic choice of the
  ridge parameter.
\newblock {\em Genet. Epidemiol.}, 37(7):704--714, 2013.

\bibitem[Dalalyan, Hebiri and Lederer (2017)]{DalalyanHebiriLederer}
Arnak~S Dalalyan, Mohamed Hebiri, and Johannes Lederer.
\newblock {On the prediction performance of the Lasso}.
\newblock {\em Bernoulli}, 23(1):552--581, feb 2017.

\bibitem[Demkow and Wola{\'n}czyk, 2017]{demkow2017genetic}
U~Demkow and T~Wola{\'n}czyk.
\newblock Genetic tests in major psychiatric disorders--integrating molecular
  medicine with clinical psychiatry--why is it so difficult?
\newblock {\em Transl. Psychiat.}, 7(6):1--9, 2017.

\bibitem[Ehret et al., 2011]{ehret2011genetic}
Georg~B Ehret, Patricia~B Munroe, Kenneth~M Rice, Murielle Bochud, Andrew~D
  Johnson, Daniel~I Chasman, Albert~V Smith, Martin~D Tobin, Germaine~C
  Verwoert, Shih-Jen Hwang, et~al.
\newblock Genetic variants in novel pathways influence blood pressure and
  cardiovascular disease risk.
\newblock {\em Nature}, 478(7367):103--109, 2011.

\bibitem[Golub, Heath and Wahba, 1979]{golub1979generalized}
Gene~H Golub, Michael Heath, and Grace Wahba.
\newblock Generalized cross-validation as a method for choosing a good ridge
  parameter.
\newblock {\em Technometrics}, 21(2):215--223, 1979.

\bibitem[Guy et al., 2010]{guy2010social}
Ido Guy, Naama Zwerdling, Inbal Ronen, David Carmel, and Erel Uziel.
\newblock Social media recommendation based on people and tags.
\newblock In {\em Proceedings of the 33rd international {ACM} {SIGIR}
  conference on Research and Development in Information Retrieval}, pages
  194--201. ACM, 2010.

\bibitem[Hamburg and Collins, 2010]{hamburg2010path}
Margaret~A Hamburg and Francis~S Collins.
\newblock The path to personalized medicine.
\newblock {\em New Engl. J. Med.}, 363(4):301--304, 2010.

\bibitem[Hebiri and Lederer (2013)]{Hebiri2013}
Mohamed Hebiri and Johannes Lederer.
\newblock How correlations influence lasso prediction.
\newblock {\em {IEEE} T.\ Inform.\ Theory}, 59(3):1846--1854, 2013.

\bibitem[Hellton and Hjort (2018)]{hellton2018fridge}
Kristoffer~H Hellton and Nils~Lid Hjort.
\newblock Fridge: Focused fine-tuning of ridge regression for personalized
  predictions.
\newblock {\em Stat. Med.}, 37(8):1290--1303, 2018.

\bibitem[Hippisley-Cox and Coupland, 2015]{hippisley2015development}
Julia Hippisley-Cox and Carol Coupland.
\newblock Development and validation of risk prediction algorithms to estimate
  future risk of common cancers in men and women: prospective cohort study.
\newblock {\em BMJ Open}, 5(3):1--25, 2015.

\bibitem[Hoerl and Kennard, 1970]{hoerl1970ridge}
Arthur~E Hoerl and Robert~W Kennard.
\newblock Ridge regression: Biased estimation for nonorthogonal problems.
\newblock {\em Technometrics}, 12(1):55--67, 1970.



\bibitem[Lederer et al., 2019]{lederer2019oracle}
Johannes Lederer, Lu~Yu, Irina Gaynanova, et~al.
\newblock Oracle inequalities for high-dimensional prediction.
\newblock {\em Bernoulli}, 25(2):1225--1255, 2019.

\bibitem[Lepskii, 1992]{lepskii_1992}
Oleg~V Lepskii.
\newblock On problems of adaptive estimation in white gaussian noise.
\newblock {\em Topics in {N}onparametric {E}stimation}, 12:87--106, 1992.


\bibitem[Nam et al., 2007]{nam2007assessing}
Robert~K Nam, Ants Toi, Laurence~H Klotz, John Trachtenberg, Michael~AS Jewett,
  Sree Appu, D~Andrew Loblaw, Linda Sugar, Steven~A Narod, and Michael~W
  Kattan.
\newblock Assessing individual risk for prostate cancer.
\newblock {\em J. Clin. Oncol.}, 25(24):3582--3588, 2007.

\bibitem[Ogino et al., 2011]{ogino2011cancer}
Shuji Ogino, J{\'e}r{\^o}me Galon, Charles~S Fuchs, and Glenn Dranoff.
\newblock Cancer immunology—analysis of host and tumor factors for
  personalized medicine.
\newblock {\em Nat. Rev. Clin. Oncol.}, 8(12):711, 2011.


\bibitem[Patel, 1998]{patel1998effect}
Mandakini~G Patel.
\newblock The effect of dietary intervention on weight gains after renal
  transplantation.
\newblock {\em J. Renal Nutr.}, 8(3):137--141, 1998.



\bibitem[Rafailidis et al., 2014]{rafailidis2014content}
Dimitrios Rafailidis, Apostolos Axenopoulos, Jonas Etzold, Stavroula
  Manolopoulou, and Petros Daras.
\newblock Content-based tag propagation and tensor factorization for
  personalized item recommendation based on social tagging.
\newblock {\em ACM Trans. Interact. Intell. Syst.}, 3(4):1--26, 2014.


\bibitem[Shao and Deng, 2012]{Shao2012}
Jun Shao and Xinwei Deng.
\newblock Estimation in high-dimensional linear models with deterministic
  design matrices.
\newblock {\em Ann. Stat.}, 40(2):812--831, 2012.

\bibitem[Spokoiny, Mammen and Lepsiki, 1997]{spokoiny_mammen_lepski_1997}
V.~G. Spokoiny, E.~Mammen, and O.~V. Lepski.
\newblock Optimal spatial adaptation to inhomogeneous smoothness: an approach
  based on kernel estimates with variable bandwidth selectors.
\newblock {\em Ann. Stat.}, 25(3):929–947, 1997.

\bibitem[Stone, 1974]{stone1974cross}
Mervyn Stone.
\newblock Cross-validatory choice and assessment of statistical predictions.
\newblock {\em J. R. Statist. Soc. B}, 36(2):111--133, 1974.

\bibitem[Tang, Liao and Sun, 2013]{tang2013prediction}
Heng Tang, Stephen~Shaoyi Liao, and Sherry~Xiaoyun Sun.
\newblock A prediction framework based on contextual data to support mobile
  personalized marketing.
\newblock {\em Decis. Support Syst.}, 56:234--246, 2013.


\bibitem[Waters et al., 2018]{waters2018human}
Shafagh Al~Nadaf Waters, Sharon Leh~Ing Wong, Nikhil~Tanaji Awatade,
  Christopher~Kenta Hewson, Laura~Katherine Fawcett, Anthony Kicic, and Adam
  Jaffe.
\newblock Human primary epithelial cell models: promising tools in the era of
  cystic fibrosis personalized medicine.
\newblock {\em Front. Pharmacol.}, 9:1--11, 2018.

\bibitem[Ziegler et al., 2012]{ziegler2012personalized}
Andreas Ziegler, Armin Koch, Katja Krockenberger, and Anika Gro{\ss}hennig.
\newblock Personalized medicine using dna biomarkers: a review.
\newblock {\em Hum. Genet.}, 131(10):1627--1638, 2012.

\end{thebibliography}

\end{document}